\documentclass[conference]{IEEEtran}

\makeatletter
\def\ps@headings{%
\def\@oddhead{\mbox{}\scriptsize\rightmark \hfil \thepage}%
\def\@evenhead{\scriptsize\thepage \hfil \leftmark\mbox{}}%
\def\@oddfoot{}%
\def\@evenfoot{}}
\makeatother
\usepackage{amsmath,amsthm}
\usepackage{amssymb}
\usepackage{fancyvrb}

\usepackage{graphicx,graphics}
\usepackage{algorithm}
\usepackage{algorithmic}
\usepackage{url}
\usepackage{balance}
\usepackage{longtable}
\usepackage{subcaption}
\usepackage{epstopdf}

\usepackage{tabularx}
\usepackage{color}
\newtheorem{thm}{Theorem}
\newtheorem{lemma}[thm]{Lemma}

\begin{document}

\title{A Robust Information Source Estimator with Sparse Observations}
\author{\IEEEauthorblockN{Kai Zhu and Lei Ying}\\
\IEEEauthorblockA{School of Electrical, Computer and Energy Engineering\\ Arizona State University\\
Tempe, AZ, United States, 85287\\
Email: kzhu17@asu.edu, lei.ying.2@asu.edu}}

\maketitle
\begin{abstract}
In this paper, we consider the problem of locating the information source with sparse observations. We assume that a piece of information spreads in a network following a heterogeneous susceptible-infected-recovered (SIR) model and that a small subset of infected nodes are reported, from which we need to find the source of the information. We adopt the sample path based estimator developed in \cite{ZhuYin_12}, and prove that on infinite trees,  the sample path based estimator is a Jordan infection center with respect to the set of observed infected nodes. In other words, the sample path based estimator minimizes the maximum distance to observed infected nodes.  We further prove that the distance between the estimator and the actual source is upper bounded by a constant independent of the number of infected nodes with a high probability on infinite trees. Our simulations on tree networks and real world networks show that the sample path based estimator is closer to the actual source than several other algorithms.

\end{abstract}
\section{Introduction}\label{sec:introduction}

In this paper, we are interested in locating the source of information that spreads in a network by using sparse observations. The solution to this problem has important applications such as locating the sources of epidemics, news/rumors in social networks or online computer virus. The problem has been studied in \cite{ShaZam_10,ShaZam_11,ShaZam_12,LuoTayLen_12} under a homogeneous susceptible-infection (SI) model and in \cite{ZhuYin_12} under a homogeneous susceptible-infection-recover (SIR) model, assuming that a complete snapshot of the network is given.

While \cite{ShaZam_10,ShaZam_11,ShaZam_12,LuoTayLen_12,ZhuYin_12} answered some fundamental questions about information source detection in large-scale networks, a complete snapshot of a real world network, which may have hundreds of millions of nodes, is expensive to obtain. Furthermore, these works assume homogeneous infection across links and homogeneous recovery across nodes, but in reality, most networks are heterogeneous. For example, people close to each other are more likely to share rumors and epidemics are more infectious in the regions with poor medical care systems. Therefore, it is important to take sparse observations and network heterogeneity into account when locating information sources. In this paper, we consider a heterogeneous SIR model and assume only a small subset of infected nodes are reported to us. The goal is to  identify the information source in a heterogeneous network by using sparse observations.

We use the sample path based approach developed in \cite{ZhuYin_12} for locating the information source with sparse observations. Surprisingly, we find that the sample path based estimator is robust to network heterogeneity and the number of observed infected nodes. In particular, our results show that even under a heterogeneous SIR model and with sparse observations, the sample path based estimator remains to be a Jordan infection center in infinite trees, where the Jordan infection center with a partial observation is the node that minimizes the maximum distance to observed infected nodes. We further show that in an infinite tree, the distance between a Jordan infection center and the actual source can be bounded by a value independent of the size of infected subnetwork with a high probability, where the infected subnetwork is the subnetwork that consists of nodes are either infected or recovered and is a connected component. Assume the size of the infected subnetwork is $n,$ the result says that a Jordan infection center is  a distance of $O(1)$ from the actual source.

We remark that the locations of the Jordan centers only depend on the network topology and are independent of the infection and recovery probabilities, so the sample path based estimators (or the Jordan infection centers) are also robust to the information diffusion models, which makes it very appealing in practice since the accurate knowledge of the SIR parameters can be difficult to measure in reality.

\subsection{Related Works}
Other than \cite{ShaZam_10,ShaZam_11,ShaZam_12,LuoTayLen_12,ZhuYin_12}, there are several related works in this area including: (1) detecting the first adopter of an innovation based on game theory \cite{SubBer_12}, in which the maximum likelihood estimator is derived but the computational complexity of finding the estimator is exponential in the number of nodes; (2) distinguishing epidemic infection from random infection under the SI model \cite{MilCarSha_12}; (3) geospatial abduction which deals with reasoning certain locations in a two-dimensional geographical area that can explain observed phenomena \cite{ShaSubSap_11,ShaSub_11}. A recent paper {\cite{LokMezOht_13} also proposed a dynamic message passing algorithm (DMP) to detect the information source under a general SIR model with complete or partial observations. However, the algorithm needs the complete information of infection and recovery probabilities. In addition, the complexity of DMP is very high under partial observations since almost all nodes in the network are candidates of the source, and the calculation needs to be repeated for every possible candidate. In the simulations, we will show that our algorithm significantly outperforms DMP in terms of both accuracy and speed. We will see that our algorithm is $400\times$ faster even when we limit the DMP algorithm to a subnetwork. 

\section{A heterogeneous SIR Model}
In this section, we introduce the heterogeneous SIR model for information propagation. Different from the homogeneous SIR model in which infection and recovery probabilities are both homogeneous \cite{ZhuYin_12}, the heterogeneous SIR model we consider allows different infection probabilities at different links and different recovery probabilities at different nodes.

Consider an undirected graph $G=\{\cal{V},\cal{E}\},$ where $\cal{V}$ is the set of nodes and ${\cal E}$ is the set of edges. Denote by $(u,v)\in{\cal E} $ the edge between node $u$ and node $v.$ Each node $v\in\cal{V}$ has three states: susceptible ($S$), infected ($I$), and recovered ($R$). Time is slotted. At the beginning of each time slot, each infected node attempts to contact all its susceptible neighbors. A contact from node $u$ to node $v$ {\em succeeds} with probability $q_{uv}.$ A susceptible node becomes infected after being {\em successfully} contacted by one of its infected neighbors. At the middle of each time slot, an infected node, {\em if it is infected before the current time slot,} recovers with probability $p_v.$ A recovered node cannot be infected again. We assume that contacts succeed independently across links and time slots; and nodes recover independently across nodes and time slots.

Consider a network shown in Figure \ref{fig: model-example}, where node $e$ is in the susceptible state, nodes $a$ and $c$ are in the infected state and nodes $b$ and $d$ are in the recovered state. Then at the next time slot, node $e$ becomes infected with probability $$1-(1-q_{ae})(1-q_{ce}),$$ and nodes $a$ and $c$ recover with probability $p_a$ and $p_c,$ respectively.
\begin{figure}[htb]
\begin{centering}
  \includegraphics[width=0.42\columnwidth]{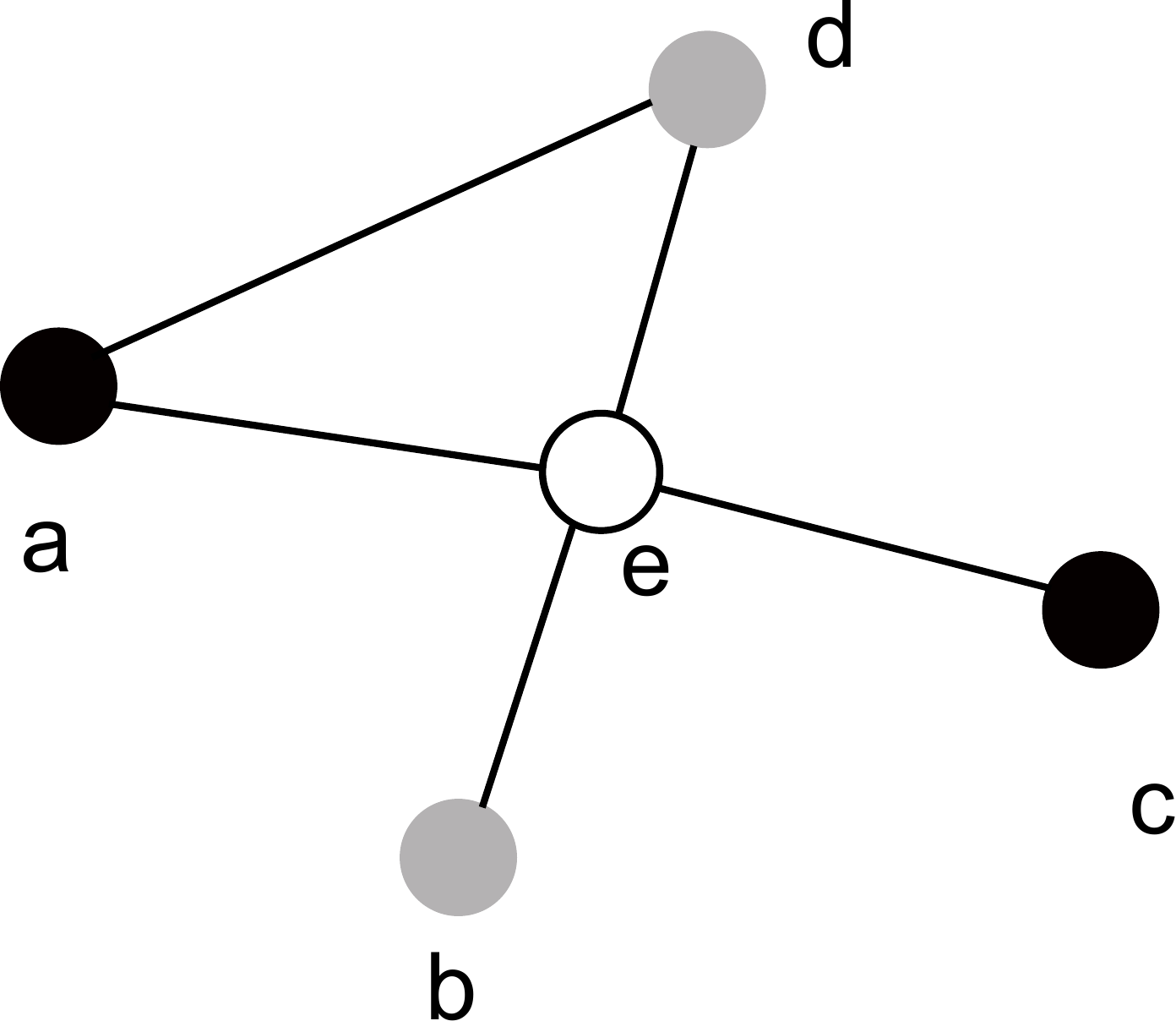}
  \caption{An example for illustrating the heterogeneous SIR model}\label{fig: model-example}
  \end{centering}
\end{figure}

\section{Problem Formulation}
In this section, we formally define the problem of information source detection. Adopting the notation in \cite{ZhuYin_12}, we define $X_v(t)$ to be the states of node $v$ at the end of time slot $t$ such that
\[
X_v(t)=\left\{
               \begin{array}{ll}
                 S, & \hbox{if $v$ is in state $S$ at time $t$;} \\
                 I, & \hbox{if $v$ is in state $I$ at time $t$;} \\
                 R, & \hbox{if $v$ is in state $R$ at time $t$.}
               \end{array}
             \right.
\]
Let ${\bf X}(t)=\{X_v(t):\forall v \in {\cal V}\}$ denote the state of all nodes in the network at time $t.$

In this paper, we assume that we only have {\em one partial snapshot} of the network, which is {\em a subset of the infected nodes.} This observation can be sparse, and details will be given in the next section. We assume that the states of other nodes are unknown. We let $Y_v$ denote the state of node $v$ in the snapshot such that
\[
Y_v=\left\{
               \begin{array}{ll}
                 1, & \hbox{if node $v$ is observed to be infected;} \\
                 0,& \hbox{otherwise.}
               \end{array}
             \right.
\] Let ${\bf Y}=\{Y_v:\forall v \in {\cal V}\}.$ We denote by $v^*$ the information source. The problem of information source detection is to locate $v^*$ based on the partial observation ${\bf Y}$ and the network topology $G.$

Due to recovery and partial observations, all nodes in the network are potential candidates of the information source. The maximum likelihood estimator of the problem is therefore computationally expensive to find as pointed out in \cite{ZhuYin_12}. In this paper, we follow the sample path based approach proposed in \cite{ZhuYin_12} to find an estimator of $v^*.$

Since ${\bf X}(t)$ is the state of the network at time $t,$ the sequence $\{{\bf X}(\tau)\}_{0\leq \tau\leq t}$ specifies the complete infection process. Therefore, we call ${\bf X}[0,t]=\{{\bf X}(\tau): 0\leq\tau\leq t\}$ a sample path.  We further define a function $F(\cdot)$ such that
\[
F(X_v(t))=\left\{
               \begin{array}{ll}
                 1, & \hbox{if $X_v(t)=I$ and $v$ is observed;} \\
                 0, & \hbox{otherwise.} \\
               \end{array}
             \right.
\] This function maps the actual state of a node to the observed state of the node.  ${\bf F}({\bf X}(t))={\bf Y}$ if and only if $F(X_v(t))=Y_v, \forall v\in{\cal V}.$
The {\cal optimal sample path} ${\bf X}^*[0,t^*]$ is defined to be the most likely sample path that results in the observed snapshot, i.e., it solves the following optimization problem:
\begin{align}
{\bf X}^*[0,t^*]={\arg\max}_{t,{\bf X}[0,t]\in {\cal X}(t)}\Pr({\bf X}[0,t]),
\end{align}
where ${\cal X}(t)=\{{\bf X}[0,t]|{\bf F}({\bf X}(t))={\bf Y}\}.$ The source associates with ${\bf X}^*[0,t^*]$ is called \emph{the sample path based estimator}. It is proved in \cite{ZhuYin_12} that the sample path based estimator on an infinite tree is a Jordan infection center under the homogeneous SIR model with a complete snapshot.  The focus of this paper is to identify the sample path based estimator under the heterogeneous SIR model with sparse observations.

\section{Main Results}\label{sec:PartialObservation}
In this section, we summarize the main results of this paper.

\subsection{Main result 1: The Jordan infection centers as the sample path based estimators}
In our theoretical analysis, we consider tree networks with infinitely many levels (or called infinite trees) to derive the sample path based estimator under the heterogeneous SIR model with a partial snapshot. Let ${\cal I}_{\bf Y}$ denote the set of observed infected nodes. We define the {\cal observed infection eccentricity} $\tilde{e}(v,{\cal I}_{\bf Y})$ of node $v$ to be the maximum distance between $v$ and any observed infected node where the distance is defined to be the shortest distance between two nodes. The Jordan infection centers of the partial snapshot are then defined to be the nodes with the minimum observed infection eccentricity. The following theorem states that on an infinite tree, the sample path based estimator is a Jordan infection center of the partial snapshot.

\begin{thm}
Consider an infinite tree and assume the partial snapshot ${\bf Y}$ contains at least one infected node. The sample path based estimator, denoted by $v^{\dag},$ is a Jordan infection center, i.e.,
\[
v^{\dag}\in \arg \max_{v\in {\cal V}} \tilde{e}(v,{\cal I}_{{\bf Y}}).
\] \hfill{$\square$}\label{th:estimator}
\end{thm}

The proof of this theorem consists of the following key steps.

1) In the first step, we focus on the sample paths originated from node $v$ (i.e., we assume node $v$ is the source). We consider two groups of sample paths: ${\cal X}_v(t)$ and ${\cal X}_v(t+1),$ where ${\cal X}_v(t)$ is the set of the sample paths  that are {\em originated from $v,$ have time duration $t,$} and {\em are consistent with the partial snapshot}, i.e., ${\bf F}({\bf X}(t))={\bf Y}$ for any ${\bf X}[0, t]\in{\cal X}_v(t).$ The set ${\cal X}_v(t+1)$ is similarly defined. We show that for any $t\geq \tilde{e}(v,{\cal I}_{{\bf Y}}),$ the sample path with the highest probability in ${\cal X}_v(t)$ occurs more likely than the one in ${\cal X}_v(t+1).$ In other words, $$\max_{{\bf X}[0,t]\in {\cal X}_v(t)} \Pr({\bf X}[0,t])> \max_{{\bf X}[0,t+1] \in {\cal X}_v(t+1)} \Pr({\bf X}[0,t+1]).$$ As a consequence of this result, we conclude that {\em the sample path that has the highest probability among those originated from node $v$ has a duration of $\tilde{e}(v,{\cal I}_{{\bf Y}})$ (the observed infection eccentricity of node $v$)}. This result will be proved in Lemma \ref{lem:timeinequalityp} in Section \ref{sec:proofs}.

2) In the second step, we consider two neighboring nodes, say nodes $u$ and $v;$ and assume node $v$ has a smaller observed infection eccentricity than node $u.$ Based on Lemma \ref{lem:timeinequalityp}, we will prove that the optimal sample path associated with node $v$ occurs with a higher probability than that of node $u.$ The key idea is to construct a sample path originated from node $v$ based on the optimal sample path originated from node $u$ and show that it occurs with a higher probability. This result will be proved in Lemma \ref{lem:mainlemmap} in Section \ref{sec:proofs}.

3) We will finally prove that starting from any node, there exists a path from the node to a Jordan infection center such that the observed infection eccentricity strictly decreases along the path. Consider an example in Figure \ref{fig: intuition-estimator}. Nodes $b$ and $f$ are two observed infected nodes. So node $a$ is a Jordan infection center with observed infection eccentricity $1.$ The path from node $e$ to node $a$ is $$e\rightarrow d \rightarrow c \rightarrow b \rightarrow a,$$ along which the observed infection eccentricity decreases as $$5\rightarrow 4\rightarrow 3\rightarrow 2\rightarrow 1.$$ By repeatedly using Lemma \ref{lem:mainlemmap}, it can be shown that the optimal sample path originated from a Jordan infection center occurs with a higher probability than the optimal sample path originated from a node which is not a Jordan infection center, which implies the sample path based estimator must be a Jordan infection center.

\begin{figure}[htb]
\begin{centering}
  \includegraphics[width=0.7\columnwidth]{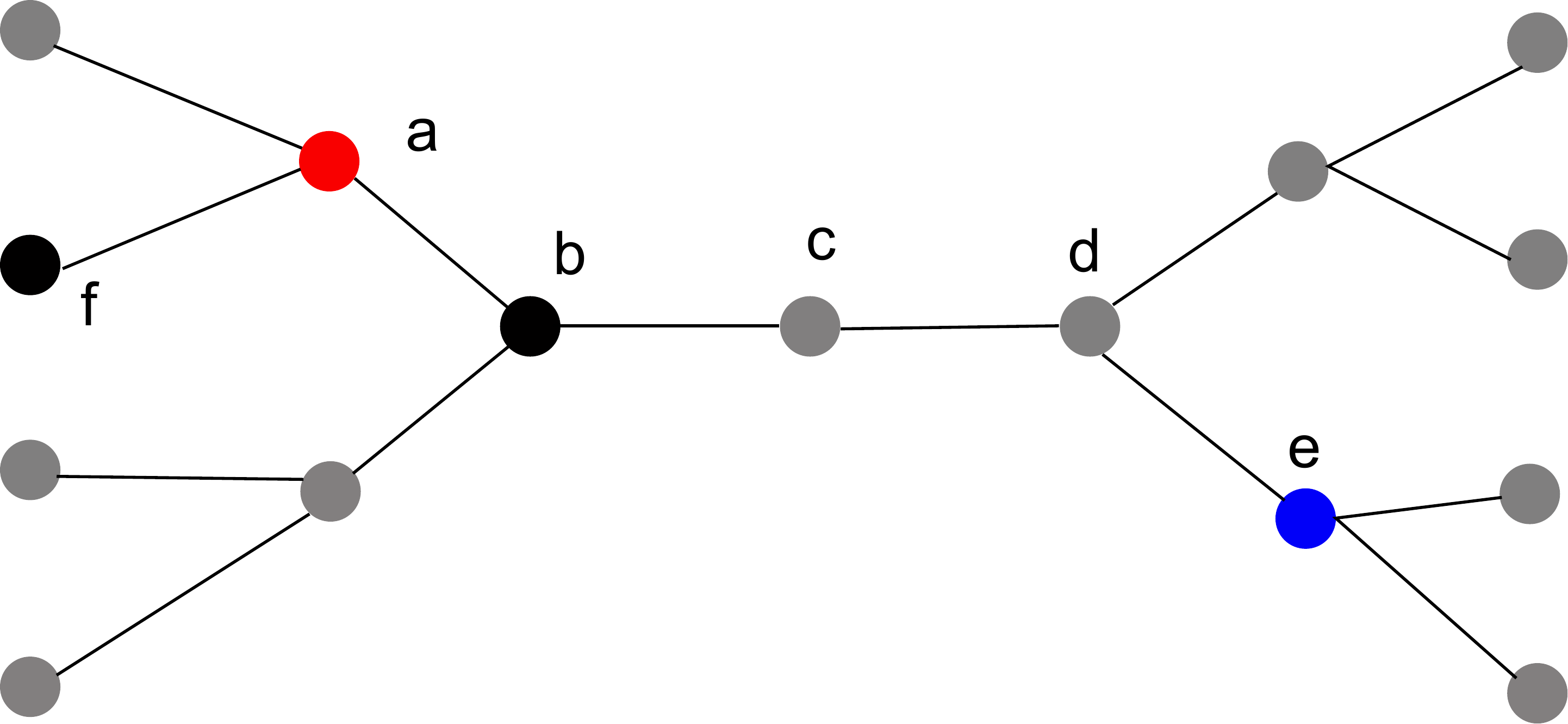}
  \caption{The key intuition behind Theorem \ref{th:estimator}}\label{fig: intuition-estimator}
  \end{centering}
\end{figure}

\subsection{Main result 2: An $O(1)$ bound on the distance between a Jordan infection center and the actual information source}
Unlike the maximum likelihood estimator, the sample path estimator is not guaranteed that the estimator is the node that most likely leads to the observation. It has been shown in \cite{ZhuYin_12} that on tree networks and under the homogeneous SIR model, the distance between the estimator and the actual source is a constant with a high probability.  It is easy to see that with a partial observation, the distance between the estimator and the actual source cannot be bounded if the observed infection nodes are arbitrarily chosen. In this paper, we consider a class of fairly general sampling algorithms that generate the partial observation  (and maybe sparse). The sampling algorithms have the following property: {\em for any set of $M$ infected nodes, the probability that at least one node in the set is reported approaches to one as $M$ goes to infinity.} We call such a sampling algorithm {\em unbiased,} in other words, any subset of infected nodes is likely to contain an observed infected node when the size of the subset is large enough. Note that if an infected node is reported with probability at least $\delta$ for some $\delta>0,$ independent of other nodes, then it satisfies the property above.  Our second main result is that the sample path estimator is within a constant distance from the actual source independent of the size of the infected subnetwork if the sampling algorithm is unbiased. {\em We also emphasize that the observation generated by an unbiased sampling algorithm can be very sparse since we only require one observed infected node is reported with a high probability among $M$ nodes when $M$ is sufficiently large.} 

\begin{thm}\label{th:perfromanceguarantee}
Consider an infinite tree. Let $g_{\min}$ be the lower bound on the number of children, and $q_{\min}>0$ be the lower bound on $q.$  Assume $g_{\min}>1,$ $g_{\min} q_{\min}>1,$ and the
observed infection topology ${\bf Y}$ contains at least one infected node and is generated by an unbiased sampling algorithm.  Then given $\epsilon>0$, the distance between the sample path estimator and the actual source is $d_\epsilon$ with
probability $1-\epsilon,$ where $d_{\epsilon}$ is independent of the size of the infected subnetwork. In other words, the
distance is $O(1)$ with a high probability. \hfill{$\square$}.
\end{thm}

\begin{figure}[htb]
\begin{centering}
  \includegraphics[width=0.99\columnwidth]{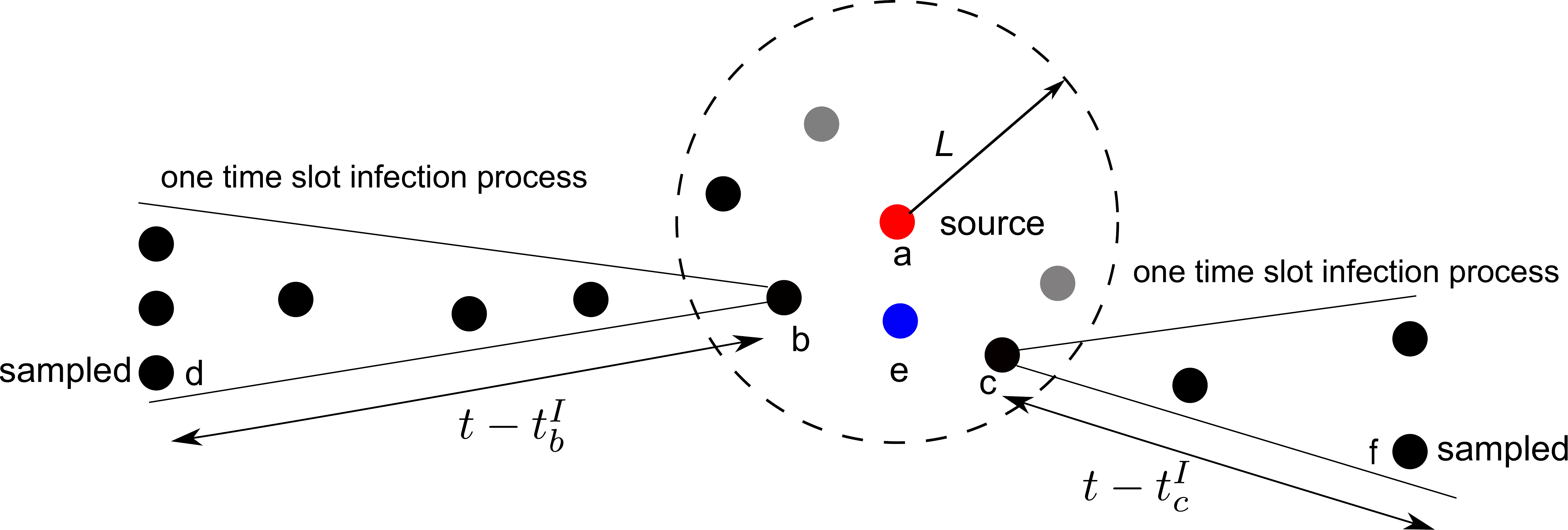}
  \caption{The key intuition behind Theorem \ref{th:perfromanceguarantee}}\label{fig: intuition-performance}
  \end{centering}
\end{figure}

The idea of the proof is illustrated using Figure \ref{fig: intuition-performance}, which consists of the following key steps:

 1) We first define a one-time-slot infection subtree to be a subtree of the infected subnetwork such that each node on the subtree is infected in the next time slot after parent is infected, except the source node. Note that the depth of a one-time-slot infection subtree grows by one deterministically until it terminates. We further say a node survives at time $t$ if it is the root of a one-time-slot infection subtree which has not terminated by time $t.$

2) In the first step, we will prove that there exist at least two survived nodes within a distance $L$ from the information source. In Figure \ref{fig: intuition-performance}, node $a$ is the information source, and nodes $b$ and $c$ are two survived nodes.

3) In the second step, we will show that with a high probability, at least one infected node at the bottom of a one-time-slot infection subtree, which has not terminated, is  observed under an unbiased sampling algorithm. In Figure \ref{fig: intuition-performance}, nodes $d$ and $f$ are two sampled nodes corresponding to the two one-time-slot infection subtrees starting from  nodes $b$ and $c,$ respectively.

4) Since a one-time-slot infection subtree grows by one deterministically at each time slot, the depth of a one-time-slot infection subtree is $t-t_k^I,$ where $k$ is the root node of the one-time-slot infection subtree. Recall that the Jordan infection centers minimize the maximum distance to observed infected nodes, so a Jordan infection center must be within a $O(1)$ distance from the two survived nodes (nodes $b$ and $c$). Considering Figure \ref{fig: intuition-performance}, we know that the actual source (node $a$) has an infection eccentricity $\leq t$ since the information can propagate at most $t$ hops at time $t.$ So the infection eccentricity of the Jordan infection centers is no more than $t$ according to the definition. Assume node $e$ in Figure \ref{fig: intuition-performance} is a Jordan infection center, then it is within a distance of $O(t)$ from nodes $d$ and $f,$ so is within a distance of $O(1)$ from nodes $b$ and $c.$  Since nodes $b$ and $c$ are no more than $L$ hops from the actual source $a,$ we can conclude that the distance between the actual source $a$ and the estimator $e$ is $O(1).$

\subsection{Reverse Infection Algorithm}
The Jordan infection centers for general graphs can be identified by the reverse infection algorithm proposed in \cite{ZhuYin_12}. In the algorithm, each observed infected node broadcasts its identity (ID) to its neighbors. All nodes in the network record the distinct IDs they received. When a node receives a new distinct ID, it records it and then broadcasts it to its neighbors. This process stops when there is a node who receives the IDs from all observed infected nodes. It is easy to verify the set of nodes who first receive all infected IDs is the set of Jordan infection centers. When there are multiple Jordan infection centers in the graph, we select the one with the maximum infection closeness centrality as the information center. The infection closeness centrality is defined as the inverse of the sum of the distances from one node to all observed infected nodes.

\subsection{Discussion: Robustness}
According to the two main results above, we know that the sample path based estimator remains to be a Jordan infection center. This is a somewhat surprising result since the locations of the Jordan infection centers are determined by the topology of the network, and are {\em independent of the parameters of the heterogeneous SIR model.} In other words, the locations of the Jordan infection centers remain the same for different SIR processes as long as the set of observed infected nodes is the same. This property suggests that the sample path based estimator is a robust estimator and can be used in the case when the parameters of the SIR model are unknown, which is a very desirable property since knowing these parameters can be difficult in practice.

In the simulations, we also consider a weighted graph with the link weights chosen proportionally according to the SIR parameters and use the weighted Jordan infection centers as the estimator. Interestingly, we will see that the performance is worse than the unweighted Jordan infection centers, which again demonstrates the robustness of the sample path based estimator. 

Furthermore, the main results hold as long as the sampling algorithm is unbiased and are independent of the number of samples. So the results are valid for sparse observations and are robust to the number of observations.

\section{Simulations}
\begin{figure*}[t!]
        \centering
        \begin{subfigure}[b]{0.34\textwidth}
                \centering
                \includegraphics[width=\textwidth]{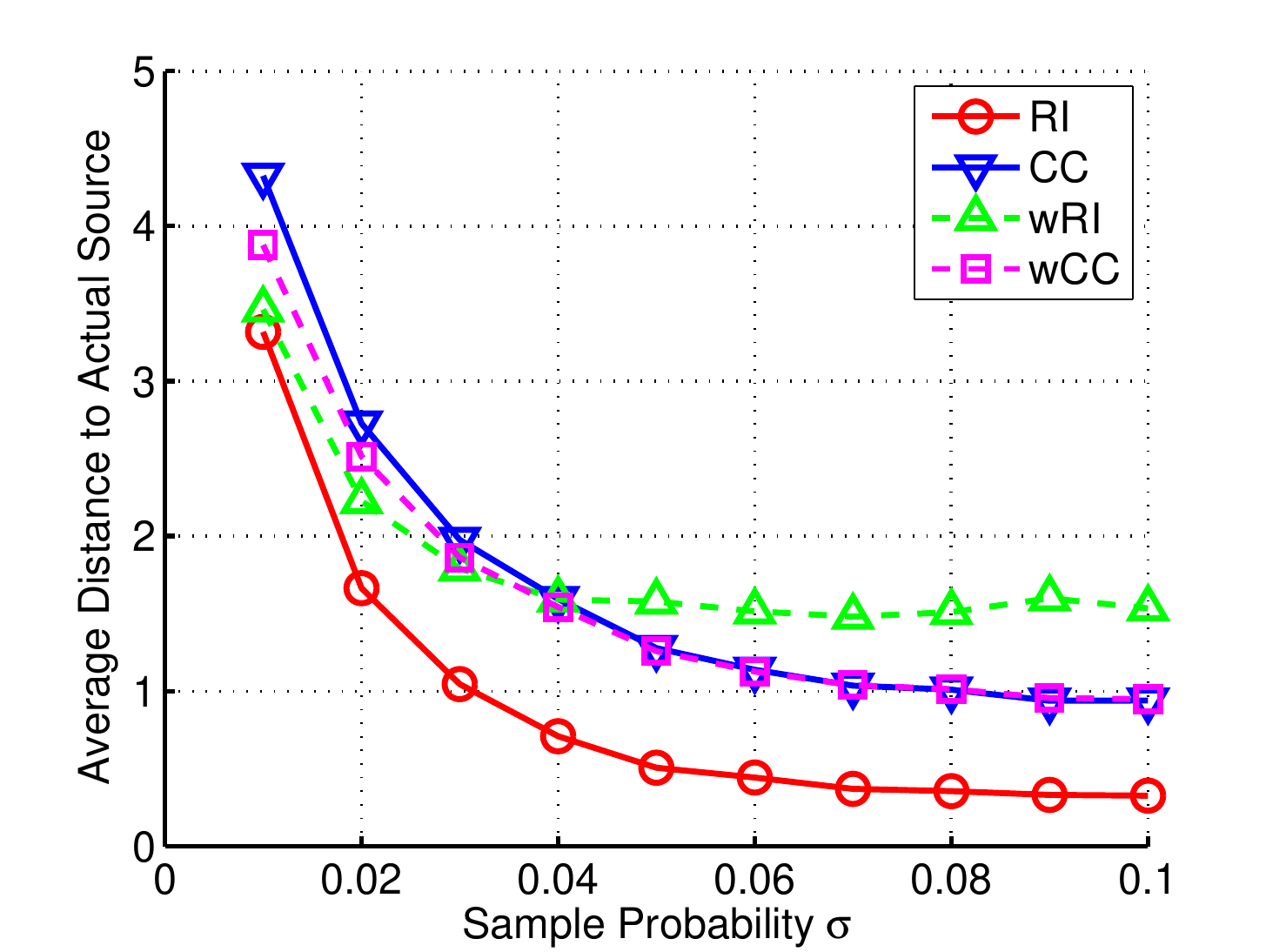}
                \caption{Regular Tree}
                \label{figure:RegularTree}
        \end{subfigure}
~
        \begin{subfigure}[b]{0.34\textwidth}
                \centering
                \includegraphics[width=\textwidth]{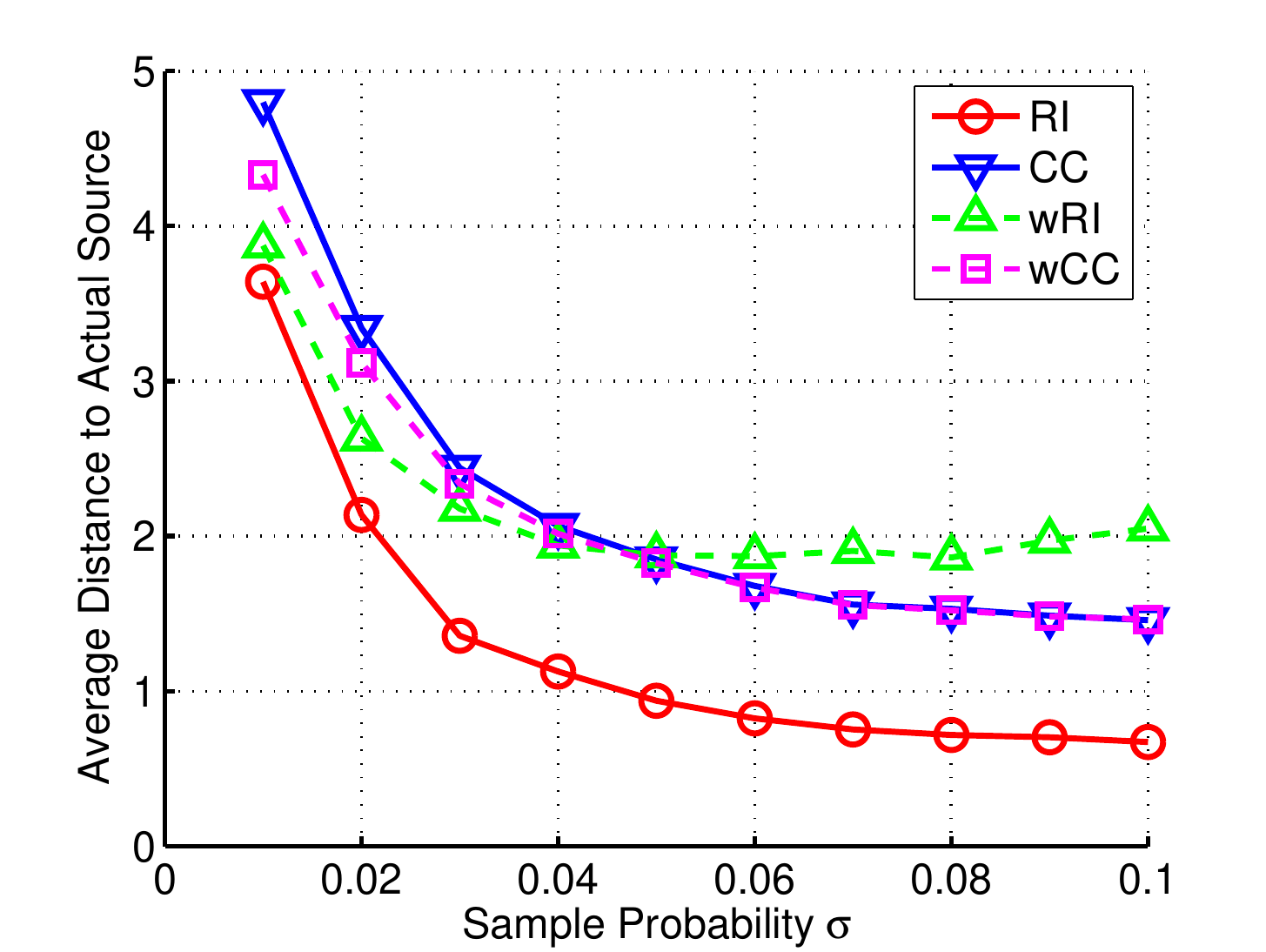}
                \caption{Binomial Tree}
                \label{figure:BinomialTree}
        \end{subfigure}\\
        \centering
        \begin{subfigure}[b]{0.34\textwidth}
                \centering
                \includegraphics[width=\textwidth]{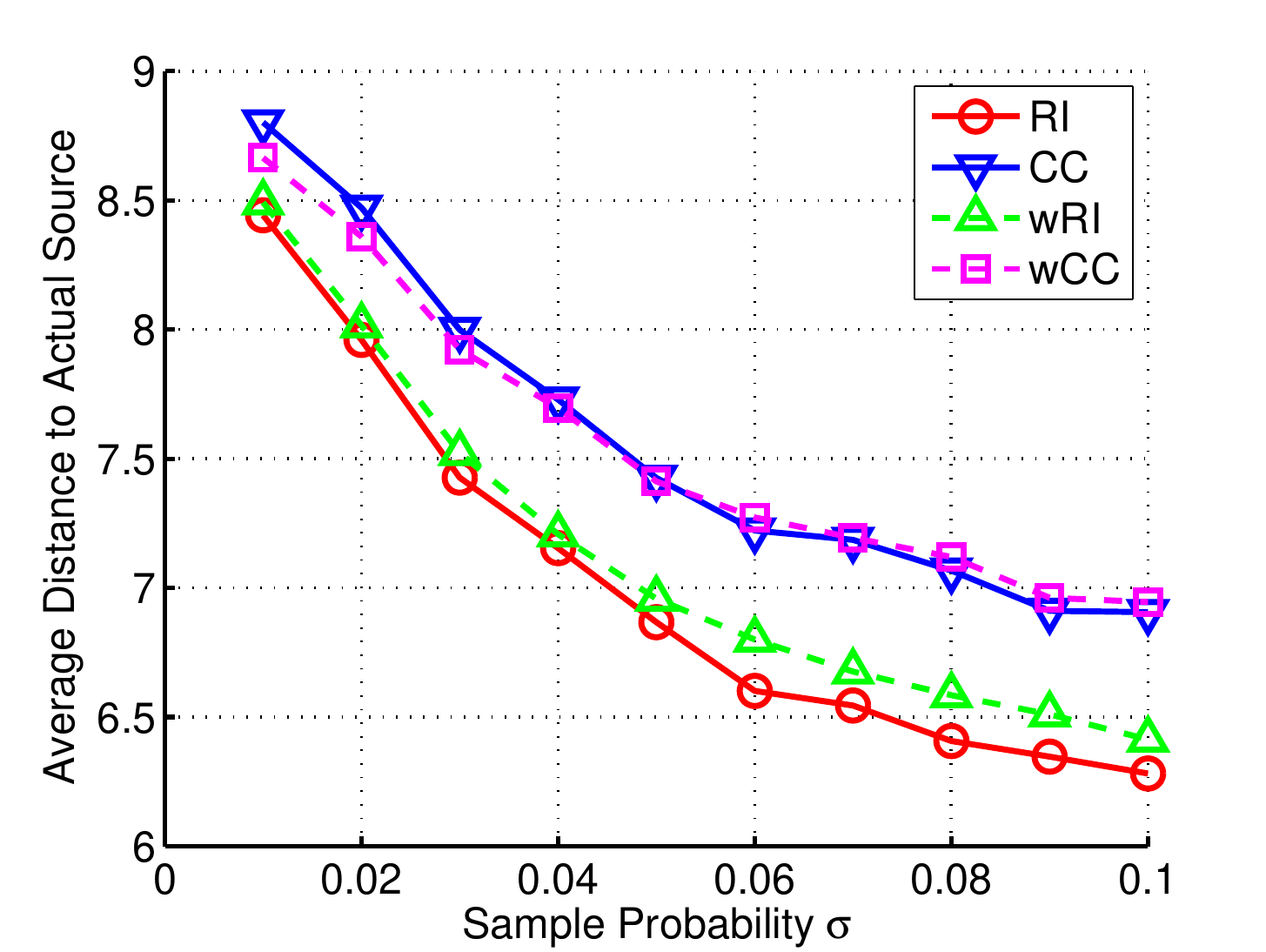}
                \caption{The Power Grid Network}
                \label{figure:PowerNetwork}
        \end{subfigure}
~
        \begin{subfigure}[b]{0.34\textwidth}
                \centering
                \includegraphics[width=\textwidth]{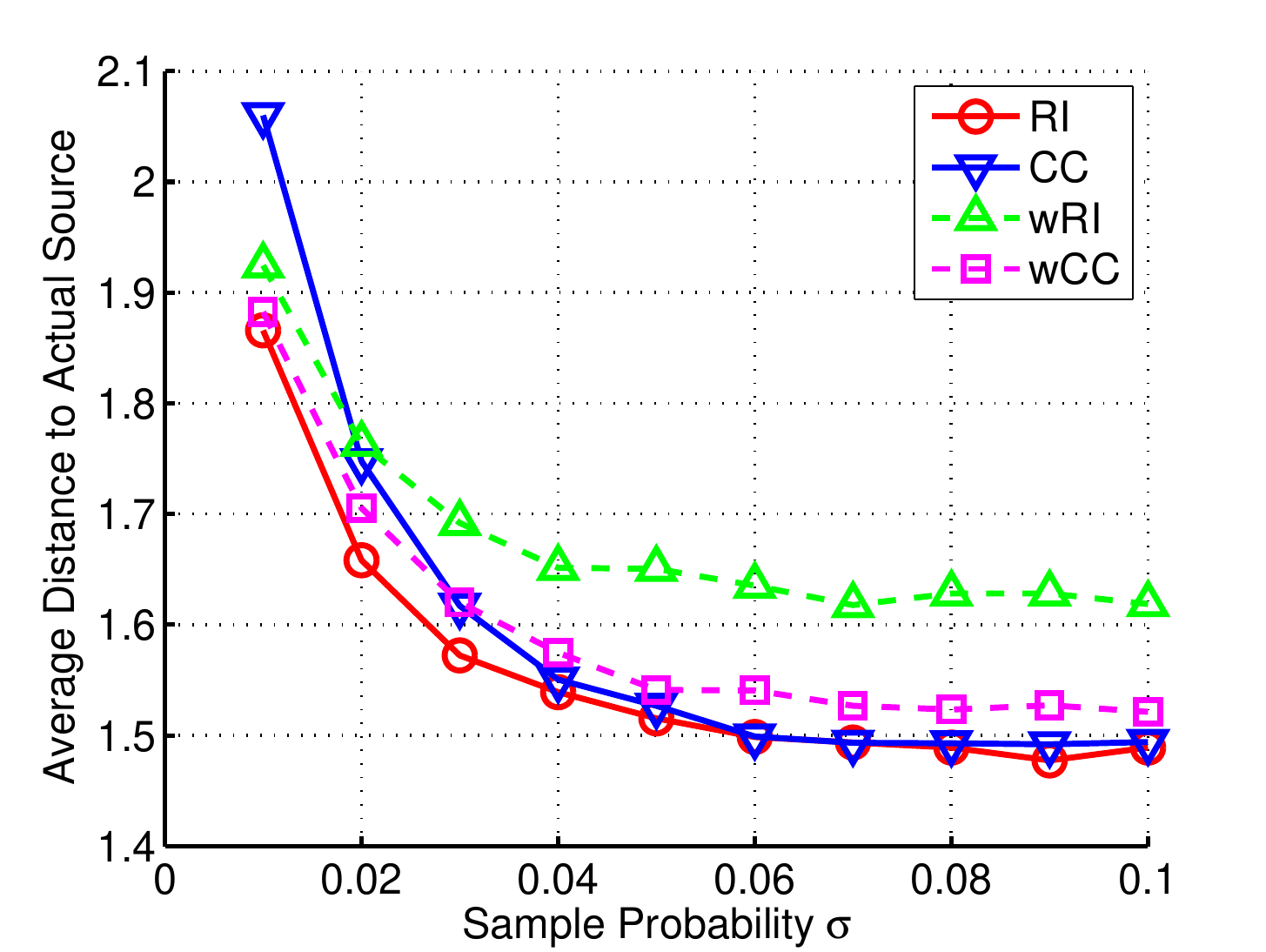}
                \caption{The Internet Autonomous Systems Network}
                \label{figure:InternetAutonomous}
        \end{subfigure}        
        \centering
\caption{\textcolor{black}{The Performance of RI, CC, wRI and wCC on Different Graphs}}
\end{figure*}

In this section, we evaluate the performance of the reverse infection algorithm for the heterogeneous SIR model on different networks including tree networks and real world networks. 

We first describe the heterogeneous SIR model we used in the simulation. Each edge $e\in\cal{E}$ is assigned with a weight $q_e$ which is uniformly distributed over $(0,1).$ The infection time over each edge $e\in\cal{E}$ is geometrically distributed with mean $1/q_e.$ Similarly, each node $v\in\cal{V}$ is assigned with a weight $p_v$ generated by an uniform distribution over $(0,1)$ and the recovery time is geometrically distributed with mean $1/p_v.$ The information source is randomly selected. The total number of infected and recovered nodes in each infection graph is within the range of $[100,300].$ Each infected node $v$ in the infection graph reports with probability $\sigma,$ independently. The snapshots used in the simulations have at least one infected node. We changed $\sigma$ and evaluated the performance on different networks.

We briefly introduce the three main algorithms which were used to compare with the reverse infection algorithm (RI).

1) \emph{Closeness centrality algorithm (CC)}: The closeness centrality algorithm selects the node with the maximum infection closeness as the information source.

2)  \emph{Weighted reverse infection algorithm (wRI)}: The weighted reverse infection algorithm selects the node with the minimum weighted infection eccentricity as the information source where the weighted infection eccentricity is similar to the infection eccentricity except that the length of a path is defined to be  the sum of the link weights instead of the number of hops, and the link weight is the average time it takes to spread the information over the link, i.e., $\left\lfloor 1/q_e\right\rfloor$ on edge $e.$ 

3) \emph{Weighted closeness centrality algorithm (wCC)}: The weighted closeness centrality algorithm selects the node with the maximum weighted infection closeness as the information source.

\subsection{Tree Networks}
We first evaluated the performance of the RI algorithm on tree networks.
\subsubsection{Regular Trees}
A $g$-regular tree is a tree where each node has $g$ neighbors. We set the degree $g=5$ in our simulations.

We varied the sample probability $\sigma$ from $0.01$ to $0.1.$ The simulation results are summarized in Figure \ref{figure:RegularTree}, which shows the average distance between the estimator and the actual information source versus the sampling probability.  \textcolor{black}{When the sample probability increases, the performance of all algorithms improve. When the sample probability is larger than $6\%,$ the average distance becomes stable which means a small number of infected nodes is enough to obtain a good estimator. We also notice that the average distance of RI is smaller than all other algorithms and is less than one hop when $\sigma\geq 0.04.$ wRI has a similar performance with RI when the sample probability is small (=0.01) but becomes much worse when the sample probability increases.}

\subsubsection{Binomial Trees}
We further evaluated the performance of RI and other algorithms on binomial trees $T(\xi,\beta)$ where the number of children of each node follows a binomial distribution such that $\xi$ is the number of trials and $\beta$ is the success probability of each trial. In the simulations, we selected $\xi=10$ and $\beta=0.4.$ Again, we varied $\sigma$ from $0.01$ to $0.1.$ The results are shown in Figure \ref{figure:BinomialTree}. \textcolor{black}{Similar to the regular trees, the performance of RI dominates CC, wRI and wCC, and the difference in terms of the average number of hops is approximately one when $\sigma\geq 0.03.$ }

\subsection{Real World Networks}
In this section, we conducted experiments on two real world networks: the Internet Autonomous Systems network (IAS)\footnote{Available at
\url{http://snap.stanford.edu/data/index.html}}, and the power grid network (PG)\footnote{Available at
\url{http://www-personal.umich.edu/~mejn/netdata/}}. 

\subsubsection{The Power Grid Network}
The power grid network has 4,941 nodes and 6,594 edges. On average, each node has  1.33 edges. So the power grid network is a sparse network. The simulation results are shown in Figure \ref{figure:PowerNetwork}. In the power grid network, we can see that RI and wRI have similar performance, and both outperform CC and wCC by at least one hop when $\sigma\geq 0.04.$

\subsubsection{The Internet Autonomous Systems Network}
The Internet Autonomous Systems network is the data collected on March, 31st, 2001. There are 10,670 nodes and 22,002 edges in the network. The simulation results are shown in Figure \ref{figure:InternetAutonomous}. \textcolor{black}{ wRI and wCC  always perform worse than RI. Although RI and CC have similar performance when the sample probability is large, RI outperforms CC when $\sigma\leq 0.03$. }

\subsection{RI versus DMP}

We finally compared the performance of RI and DMP with sparse observations. We conducted the simulation on the power grid network and fixed the sample probability to be $10\%.$  Under this setting, the complexity of DMP is very high since the DMP computation needs to be repeated for every node in the network. Since nodes far away from the observed infected nodes are not likely to be the information source, we ran DMP over a subset of nodes close to the Jordan infection centers to reduce the complexity of the algorithm.

We tested the speed of RI and DMP on a machine with 1.8 GB memory, 4 cores 2.4 GHz Intel i5 CPU and Ubuntu 12.10. \textcolor{black}{The algorithms are implemented in Python 2.7.} On average, it took RI 0.57 seconds to locate the estimator for one snapshot and took DMP 229.12 seconds. So RI is much faster than DMP.

 Figure \ref{fig: RI2DMP} shows the CDF of the distance from the estimator to the actual source under DMP and RI. We can see that RI dominates DMP, in particular, $71\%$ of the estimators under RI are no more than $7$ hops from the actual source comparing to $57 \%$ under DMP. Therefore, RI outperforms DMP in terms of both speed and accuracy. We remark that we did not compare the performance of RI and DMP on the Internet Autonomous System (IAS) network because the complexity of running DMP on a large size network like the IAS network is prohibitively high.

\begin{figure}[htb]
\begin{centering}
  \includegraphics[width=0.8\columnwidth]{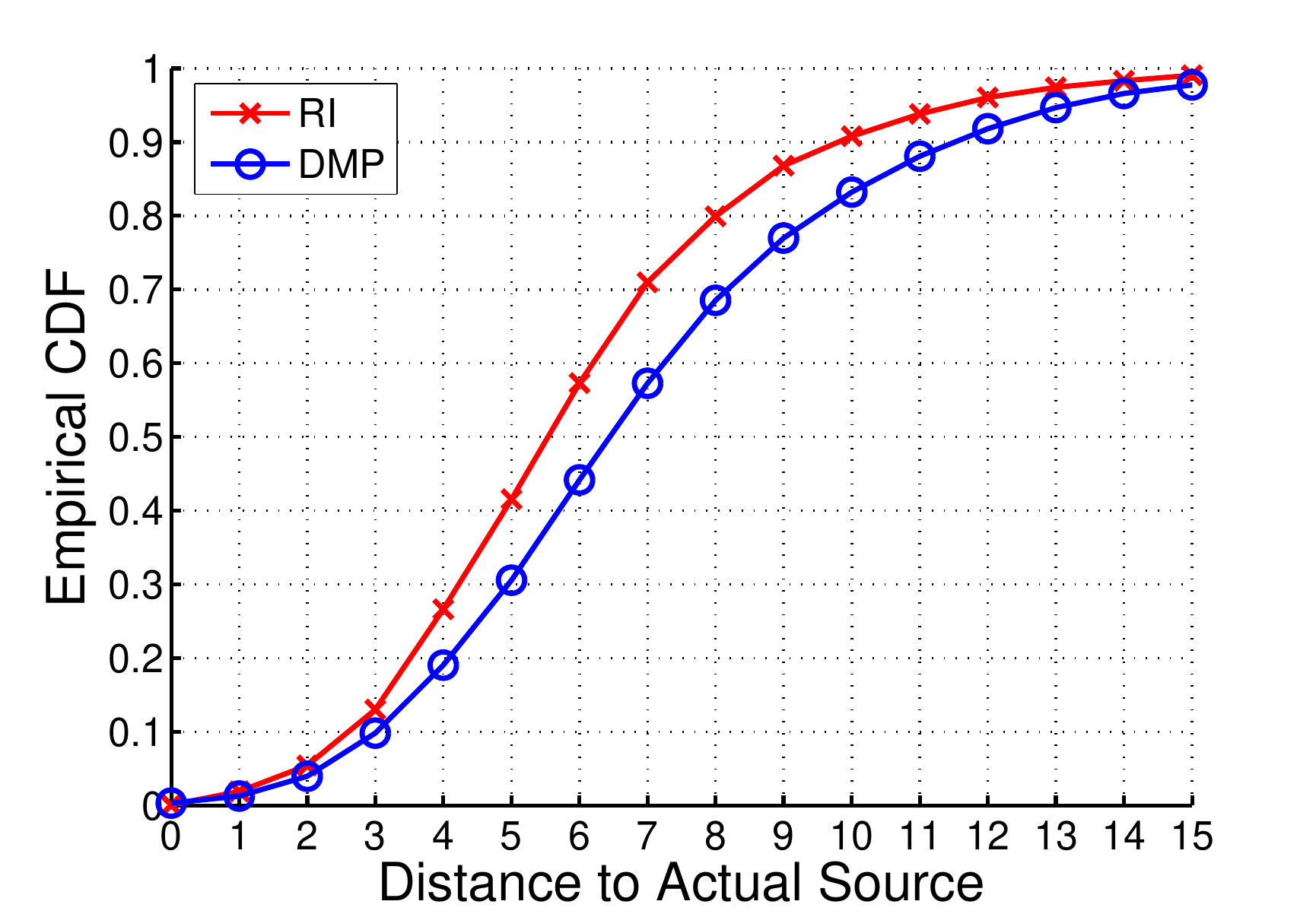}
  \caption{The CDF of RI and DMP on the Power Grid Network}\label{fig: RI2DMP}
  \end{centering}
\end{figure}

\section{Proofs}\label{sec:proofs}
In this section, we present the proofs of the main results.
\subsection{Proof of Theorem \ref{th:estimator}}
Denote by ${\cal I}_{{\bf Y}}=\{v|Y_v=1\}$ the set of observed infected
nodes and ${\cal H}_{{\bf Y}}=\{v|Y_v=0\}$ the set of unobserved nodes.
Given a node $v,$ define the optimal time $t^*_{v}$ to be
\[
t^*_{v}\triangleq\arg_t\max_{t,
{\bf X}[0,t]\in {\cal X}(t)}\Pr\left({\bf X}[0,t]|v \hbox{ is information source}\right),
\]
i.e., it is the duration of the optimal sample path with node $v$ as the information source.

\begin{lemma}{\bf (Time Inequality) }\label{lem:timeinequalityp}
Consider an infinite tree rooted at $v_r.$ Assume that $v_r$ is the information source and the observed
snapshot ${\bf Y}$ contains at least one infected node. If  $\tilde{e}(v_r,{\cal I}_{{\bf Y}})\leq t_1<t_2$, the following inequality holds,
$$\max_{{\bf X}[0,t_1]\in \tilde{{\cal X}}(t_1)} \Pr({\bf X}[0,t_1]) > \max_{{\bf X}[0,t_2] \in \tilde{{\cal X}}(t_2)} \Pr({\bf X}[0,t_2]),$$
where $\tilde{{\cal X}}(t)=\{{\bf X}[0,t]|{\bf Y}= {\bf F}({\bf X}(t))\}.$
In addition,
$$t_{v_r}^*=\tilde{e}(v_r,{\cal I}_{{\bf Y}})=\max_{u\in{\cal I}_{\bf Y}}d(v_r,u),$$
i.e., $t_{v_r}^*$ is equal to the observed infection eccentricity of $v_r$
with respect to ${\cal I}_{{\bf Y}}.$
\end{lemma}

\begin{proof}
We adopt the notations defined in \cite{ZhuYin_12}, which are listed below:
\begin{itemize}
\item ${\cal C}(v)$ is the set of children of $v.$
\item \textcolor{black}{$\phi(v)$ is the parent of node $v.$}
\item ${\cal Y}^k$ is the set of infection topologies where the maximum distance from $v_r$ to an infected node is $k$.  All possible infection topologies are then partitioned into countable subsets $\{{\cal Y}^k\}.$
\item $T_v$ is the tree rooted in $v.$
\item $T_v^{-u}$ is the tree rooted in $v$ without the branch from its neighbor $u.$
\item ${\bf X}([0,t],T_v^{-u})$ is the sample path restricted to topology $T_v^{-u}.$
\end{itemize}
Considering the case where the time difference of two sample paths is
one, we will show that
\begin{align*}
\max_{{\bf X}[0,t]\in \tilde{{\cal X}}(t)} \Pr({\bf X}[0,t])>
\max_{{\bf X}[0,t+1] \in \tilde{{\cal X}}(t+1)} \Pr({\bf X}[0,
t+1]).
\end{align*}

\textcolor{black}{Next, we use induction over ${\cal Y}^k.$ }

{\bf Step 1} ${\bf k=0}$ $v_r$ is  the only observed infected node in this case. Given a sample path ${\bf
X}[0,t+1]\in \tilde{{\cal X}}(t+1)$, the probability of the sample path can be written as
\[
\Pr\left({\bf X}[0,t+1]\right)=\Pr\left({\bf X}[0,t]\right)\Pr({\bf
X}(t+1)|{\bf X}[0,t])
\]
Since $v_r$ is the only observed infected node and all other nodes' states are unknown, we assign ${\bf X}'[0,t]\in \tilde{{\cal X}}(t)$ to be same as the first $t$ time slots in ${\bf X}[0,t+1],$ i.e., ${\bf X}'[0,t]={\bf X}[0,t].$
Hence, we obtain that
\[
\Pr\left({\bf X}'[0,t]\right)=\Pr\left({\bf X}[0,t]\right)>
\Pr\left({\bf X}[0,t+1]\right)
\]
Therefore, the case $k=0$ is proved.

{\bf Step 2} Assume the inequality holds for $k\leq n,$ and
consider $k=n+1,$ i.e., ${\bf Y}\in{\cal Y}^{n+1}.$  Clearly, $t\geq n+1\geq 1$ for each ${\bf X}[0,t].$
Furthermore, the set of subtrees ${\cal T}=\{T^{-v_r}_{u}|u \in
{\cal C}(v_r)\}$ are divided into two subsets: $${\cal
T}^h=\{T^{-v_r}_{u}|u \in {\cal C}(v_r), T^{-v_r}_{u}\cap {\cal I}_{{\bf Y}}=\emptyset\},$$ and $${\cal
T}^i={\cal T}\backslash {\cal T}^h.$$ Given  $t_{v_r}^R,$ the infection processes on the
sub-trees are mutually independent.
We construct ${\bf X}'[0,t]$ which occurs more likely than ${\bf X}^*[0,t+1]$ according to the following steps, where ${\bf X}^*[0,t+1]=\max_{{\bf X}[0,t+1] \in \tilde{{\cal X}}(t+1)} \Pr({\bf X}[0,
t+1]).$ 

{\bf Part 1 }${\bf {\cal T}^i}.$  For a subtree in ${\cal
T}^i,$ the proof follows Step 2.b and Step 2.c of Lemma 1 in \cite{ZhuYin_12}. The intuition is as follows: Consider a subtree and a sample path on it with duration $t+1.$ If $u$ is not infected at the first time slot, we can construct a sample path with duration $t$ by moving the events one time slot earlier. The new sample path (with duration $t$) has a higher probability to occur than the original one. If $u$ is infected in the first time slot, we can invoke the induction assumption to the subtree rooted at $u,$  which belongs to ${\cal Y}^n.$

{\bf Part 2} ${\bf v_r}.$ In this part, we have the freedom to assign the unobserved node as infected or healthy. In Part 1, the infection time of each
root $u$ in subtrees ${\cal T}^i$ of ${\bf X}'[0,t]$ is either
the same as or one time slot earlier than its infection time in
${\bf X}^*[0,t+1].$ Therefore, if $t^R_{v_r}\leq t,$ the recovery time of the source $v_r$ in ${\bf X}'[0,t]$ can be assigned the same as that in
${\bf X}^*[0,t+1].$

If $t^R_{v_r}=t+1,$ the source
$v_r$ recovers at time slot $t+1$ which means $v_r$ is not observed since the observation set only contains infected nodes. Therefore, in ${\bf X}'[0,t]$ we assign the
source to be in state $I$ at time $t,$ which is the same as the state of $v_r$ at time $t$ in ${\bf X}^*[0,t+1].$

If $t^R_{v_r}>t+1,$ $v_r$ remains infected in the sample path ${\bf X}^*[0,t+1].$ We assign the source to be in state $I$ in ${\bf X}'[0,t].$

As a summary, according to the assignment above, the states of the source $v_r$ in ${\bf X}'[0,t]$ are the same as those of the first $t$ time slots in ${\bf
X}^*[0,t+1].$

{\bf Part 3} ${\bf {\cal T}^h}.$  Based on the conclusion of Part 2, the subtrees belonging to ${\cal T}^h$ in ${\bf X}'[0,t]$ mimic the behaviors of  the first $t$ time slots in ${\bf
X}^*[0,t+1].$

Since ${\bf X}^*[0,t+1]$ has one extra time
slot during which some extra events occur,  ${\bf X}'[0,t]$ occurs with a higher
probability on the subtrees in ${\cal T}^h.$

According to the discussion above,
we conclude that
time inequality holds for $k=n+1,$ hence for any $k$
according to the principle of induction.
Therefore, the lemma holds.
\end{proof}

\begin{lemma} {\bf (Adjacent Nodes Inequality)}  \label{lem:mainlemmap}
Consider an infinite tree with partial observation
${\bf Y}$ which contains at least one infected node. For $u,v \in {\cal V}$ such that $(u,v)\in {\cal E}$,
if $t^*_u>t^*_v$
\begin{align*}
\Pr({\bf X}^*_u[0,t^*_u]) < \Pr({\bf X}^*_v[0,t^*_v]),
\end{align*}
where ${\bf X}^*_u[0, t_u^*]$ is the optimal sample path associated
with root $u.$
\end{lemma}

\begin{proof}
The proof of the lemma follows the proof of Lemma 2 in \cite{ZhuYin_12}. The key idea is to construct a sample path rooted at $v,$ which has a higher probability than the optimal sample path rooted at $u.$ It is not hard to see that $t^*_u=t^*_v+1$ based on the definition of the infection eccentricity. The graph is partitioned into $T^{-u}_v$ and $T^{-v}_u$ which are mutually independent after the infection of $v$ and $u.$ With this observation, we construct $\tilde{{\bf X}}_v[0,t^*_v])$ which infects $u$ at the first time slot. $\tilde{{\bf X}}_v([0,t^*_v],T^{-u}_v)$ then mimics the behavior of ${\bf X}^*_u([0,t^*_u],T^{-u}_v)$  and $\tilde{{\bf X}}_v([0,t^*_v-1],T^{-v}_u)$ has a higher probability than ${\bf X}^*_u([0,t^*_u],T^{-v}_u)$ based on Lemma \ref{lem:timeinequalityp}.
\end{proof}
The adjacent nodes inequality results in partial orders in the tree and makes it possible to compare the likelihood of optimal sample paths associated with adjacent nodes without knowing the actual probability of the optimal sample path. Following the proof of Theorem 4 in \cite{ZhuYin_12}, it can be shown that in tree networks, from any node, there exists a path from the node to a Jordan infection center such that the observed infection eccentricity strictly decreases along the path. By repeatedly using Lemma \ref{lem:mainlemmap}, we can then prove that the source of the optimal sample path must be a Jordan infection center.

\subsection{Proof of Theorem \ref{th:perfromanceguarantee}}\label{sec:performance}

In this subsection, we present the proof that shows that the sample path estimator is within a constant distance from the actual source independent of the size of the infected subnetwork. Given a tree rooted in $v^*$ where the information starts from $v^*$ following the general SIR model, we define the following three branching processes.

1) ${\cal Z}_l(T_{v^*})$ denotes the set of nodes which are in infected or recovered states at level $l$ on tree $T_{v^*}.$ Let $Z_l(T_{v^*})$ denote the cardinality of ${\cal Z}_l(T_{v^*}).$ Note that ${\cal Z}_0(T_{v^*})=\{v^*\}.$ We call this process the \emph{original infection process.}

2) ${\cal Z}^\tau_l(T_{v^*})$ denotes the set of infected and recovered
nodes at level $l$ whose parents are in set ${\cal Z}^\tau_{l-1}(T_{v^*})$ and who were infected within $\tau$ time slots after their parents were infected. This process adds a deadline $\tau$ on infection. If a node is not infected within $\tau$ time slots after its parent is infected, it is not included in this branching process. This process is called \emph{$\tau-$deadline infection process.} From the definition, if $u,v\in{\cal Z}^\tau_l(T_{v^*}),$ then
\[
|t^I_u-t^I_v|\leq l(\tau-1).
\]
For $\tau=1,$ we call ${\cal Z}^1_l(T_{v^*})$ the  \emph{one time slot infection process}. The extinction probability of a branching process is the probability that there is no offspring at certain level of the branching process, i.e., $ Z^1_l(T_{v^*})=0$ for some $l.$ Denote by $\rho_v$ the extinction probability of $Z^1_l(T^{-\phi(v)}_{v}).$

3) We define the \emph{binomial branching process} as a branching process whose offspring distribution follows binomial distribution $B(g,\varphi)$ where $g$ is the number of trials and $\varphi$ is the success probability. Denote by $\rho$ the extinction probability of the binomial branching process.

 The following notations will be used in later analysis.
\begin{itemize}
\item $v^\dag$ denotes the optimal sample path estimator.
\item $g_{\min}$ is the lower bound on the number of children, i.e.,
\[
\min_v|{\cal C}(v)|\geq g_{\min}, \forall v \in {\cal V}.
\]
\item $q_{\min}$ is the lower bound on the infection probability, i.e.,
\[
q_{\min}=\min_e q_e , \forall e \in {\cal E}.
\]

\item $\sigma^\tau_v$ is the probability that a node $v$ infects at least one of its children within $\tau$ time slot after $v$ is infected.

\end{itemize}

Given $n_0>0$ and $\tau>0,$ define $l^\dag=\min l'$ where $Z^\tau_{l'}(T_{v^*})>n_0,$ i.e., $l^\dag$ is the first level where the $\tau$-deadline infection process has more than $n_0$ offsprings.

Given $\tau$ and level $L\geq 2$, we consider the following two events:

 {\bf Event 1}: $Z_L(T_{v^*})=0.$
 
{\bf Event 2}: $l^\dagger\leq L$ and at least two one time slot infection processes starting from level $l^\dag$ survive, i.e., $\exists u,v\in {\cal Z}^\tau_{l^\dag}(T_{v^*})$ such that $\forall l,$ $Z^1_l(T_u^{-\phi(u)})\neq 0$ and $ Z^1_l(T_v^{-\phi(v)})\neq 0.$ In addition, at least one infected node at the bottom of each survived one time slot infection process is observed.

For event 1, no node at level $L$ gets infected and the infection
process terminates at level $L-1.$ So the infection eccentricity of $v^*$ is at most $L-1,$ and the minimum infection eccentricity of the network is at most $L-1.$ Therefore, the distance between
$v^*$ and $v^{\dag}$ is no more than $2(L-1).$

Considering event 2,  we assume the information propagates for $t$ time slots. The deadline property of the $\tau$-deadline infection process indicates $t^I_{u_1}\leq\tau l^\dag$ and $t^I_{u_2}\leq\tau l^\dag.$  Given a node $\tilde{v}$ at level $(\tau+1)l^\dag-1$ where  $\tilde{v}\in T_{u_2}^{-\phi(u_2)}$ and a node $v'\in T_{u_1}^{-\phi(u_1)}$ which is an observed infected node at the bottom of the infection tree, from Figure \ref{figure:PerformanceDistance}, we obtain
\begin{align*}
d(\tilde{v},v')&=t-t^I_{u_1}+\tau l^\dag+1\\
&\geq t+1.
\end{align*}
Note that $\forall u \in {\cal I},$
\[
d(v^*,u)\leq t<d(\tilde{v},v').
\]
Since $l^\dag\leq L,$ any node at or below level $L(\tau+1)-1$ has an infection eccentricity larger than that of $v^*.$ Hence, $v^\dag$ cannot be at or below level $L(\tau+1)-1.$ Therefore,
$$d(v^\dag,v^*)< (\tau+1)L-1.$$

Next, we prove the probability that either event 1 or event 2 happens goes asymptotically to $1.$ Denote by $K_{l^\dag}$ the number of one time slot infection processes which start from level $l^\dag$ and survive. Denote by $E$ the event that a  survived one time slot infection process has at least one observed infected node at its lowest level.
\begin{figure}[t!]
\begin{centering}
  \includegraphics[width=0.7\columnwidth]{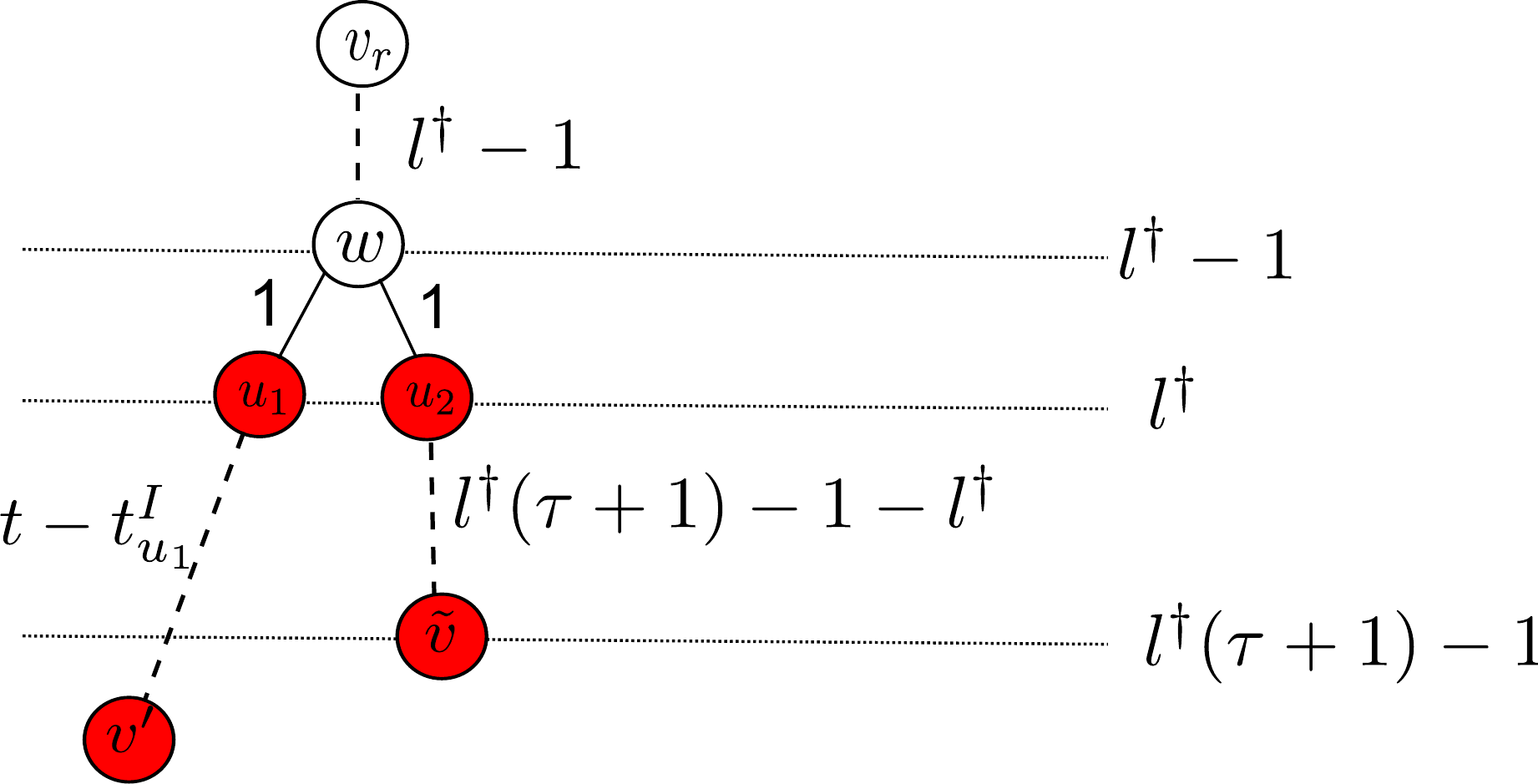}\\
  \caption{A Pictorial Description of the Distance Relations in Theorem \ref{th:perfromanceguarantee}}\label{figure:PerformanceDistance}
  \end{centering}
\end{figure}
According to the discussion above, the probability that the distance between the estimator and the actual source is no more than $(\tau+1)L-1$ is at least
\begin{align*}
&\Pr(Z_L(T_{v^*})=0)+\Pr(K_{l^\dag}\geq 2, l^\dag\leq L)\Pr(E)^2\\
&\geq\Pr(Z_L(T_{v^*})\!=\!0)\!+\!\Pr\left(l^\dag\!\leq \!L\right) \Pr\left(K_{l^\dag}\!\geq \!2\Big|l^\dag\!\leq \!L\right)\!\Pr(E)^2\\
&=\Pr(Z_L(T_{v^*})=0)+\Pr\left(\bigcup_{i=1}^{L}
Z^{\tau}_i>n_0\right)\\
&\times \Pr(K_{l^\dag}\geq 2|l^\dag\leq L)\Pr(E)^2\\
&=\left(1\!\!-\!\!\Pr\left(\bigcap_{i=1}^{L}0<Z^{\tau}_i(T_{v^*})\leq n_0\right)\!\!-\!\Pr\left(\bigcup_{i=1}^{L}Z^{\tau}_i(T_{v^*})\!=\!0\right)\right)\\
&\times\Pr(K_{l^\dag}\geq 2|l^\dag\leq L)\Pr(E)^2+\Pr(Z_L(T_{v^*})=0).
\end{align*}
In addition, we have
\begin{align}
&\Pr(K_{l^\dag}\geq 2|l^\dag\leq L)\\
&=\sum_{l=1}^L \Pr(K_{l^\dag}\geq 2, l^\dagger=l| l^\dag\leq L)\\
&=\sum_{l=1}^L \Pr(K_{l^\dag}\geq 2|  l^\dagger=l) \Pr( l^\dagger=l|l^\dag\leq L). \label{eqn:KL}
\end{align}
In Lemma 5, we prove  that the extinction probability of each branching process from level $l^\dag$ is upper bounded by the exitinction probability $\rho$ of the binomial infection process $B(g_{\min},q_{\min}).$  Therefore, at level $l^\dag$ we have $n_0$ i.i.d one time infection processes whose extinction probabilities are upper bounded by $\rho.$ The probability that at least two of them survive goes asymptotic to 1 when $n_0$ increases. Therefore, $\forall \epsilon_1>0,$ we have enough large $n_0,$ such that
\[
\Pr(K_{l^\dag}\geq 2|  l^\dagger=l)\geq 1-\epsilon_1.
\]
Therefore, equation (\ref{eqn:KL}) becomes
\begin{align*}
&\Pr(K_{l^\dag}\geq 2|l^\dag\leq L)\\
&\geq (1-\epsilon_1)\sum_{l=1}^L \Pr( l^\dagger=l|l^\dag\leq L)\\
&=(1-\epsilon_1).
\end{align*}

We show in Lemma 7 that $\Pr(E)\geq 1-\epsilon_2$ given $\epsilon_2>0.$  If
$n_0$ and $t$ are sufficiently large, we have
\[
\Pr(K_{l^\dag}\geq 2|l^\dag\leq L)\Pr(E)^2 \geq
(1-\epsilon_1)(1-\epsilon_2)^2.
\]
Therefore,
\begin{align}
&\Pr(Z_L(T_{v^*})=0)+\Pr(K_{l^\dag}\geq 2, l^\dag\leq L)\Pr(E)^2 \\
&\geq \left(1-\Pr\left(\bigcap_{i=1}^{L}0<Z^{\tau}_i(T_{v^*})\leq
n_0\right)\right)\\
&\times(1-\epsilon_1)(1-\epsilon_2)^2\\
&-\Pr\left(\bigcup_{i=1}^{L}Z^{\tau}_i(T_{v^*})=0\right)+\Pr(Z_L(T_{v^*})=0)\\
&=\underbrace{\left(1-\Pr\left(\bigcap_{i=1}^{L}0<Z^{\tau}_i(T_{v^*})\leq
n_0\right)\right)}_{\hbox{Part 1}}\\
&\times (1-\epsilon_1)(1-\epsilon_2)^2\\
&+\underbrace{\Pr(Z_L(T_{v^*})=0)-\Pr(Z_L^\tau(T_{v^*})=0)}_{\hbox{Part 2}},\label{eqn:boundedevent}
\end{align}
where equation (\ref{eqn:boundedevent}) holds since
$Z_l^{\tau}(T_{v^*})=0$ implies that $Z_L^{\tau}(T_{v^*})=0$ for $l\leq L.$

For part 1 in equation (\ref{eqn:boundedevent}), we prove in Lemma 6, given $\epsilon_3>0,$ when $\tau$ and $L$ are sufficiently large,
$$1-\Pr\left(\bigcap_{i=1}^{L}0<Z^{\tau}_i(T_{v^*})\leq
n_0\right)>1-\epsilon_3.$$

For part 2 in equation (\ref{eqn:boundedevent}), we have
\[
\lim_{\tau \rightarrow \infty}\Pr(Z_L^\tau(T_{v^*})=0)=\Pr(Z_L(T_{v^*})=0).
\]
Therefore, given $\epsilon_4>0,$ when $\tau$ is sufficiently large,
\[
\Pr(Z_L(T_{v^*})=0)-\Pr(Z_L^\tau(T_{v^*})=0)\geq -\epsilon_4.
\]

Hence, we have
\begin{align*}
&\Pr(Z_L(T_{v^*})=0)+\Pr(K_{l^\dag}\geq 2, l^\dag\leq L)\Pr(E)^2\\
&\geq(1-\epsilon_1)(1-\epsilon_2)^2(1-\epsilon_3)-\epsilon_4.
\end{align*}
 Now choosing $\epsilon_1=\epsilon_2=\epsilon_3=\epsilon_4=\epsilon_5/5$ for some $\epsilon_4>0,$ we have
\begin{align*}
&\Pr(Z_L(T_{v^*})=0)+\Pr(K_{l^\dag}\geq 2, l^\dag\leq L)\Pr(E)^2\\
&\geq 1-\epsilon_5.
\end{align*}
Now let $|{\bf Y}|$ denote the number of infected nodes in the observation $\bf Y.$ Define events $E_1=\{Z_L=0\}$ and $E_2=\{K_{l}\geq 2\hbox{ for some }l\leq L\}$ and $E_3$ is the event that two of the survived one time slot infection processes have at least one observed infected node each at their bottoms. We have
\begin{align*}
&\Pr(E_1||{\bf Y}|\geq 1)+\Pr\left(E_2\cap E_3||{\bf Y}|\geq 1\right)\\
=&\frac{1}{\Pr(|{\bf Y}|\geq 1)}\left({\Pr(E_1\cap \{|{\bf Y}|\geq 1\})}\right.\\
&\left.+{\Pr\left(E_2\cap E_3\cap \{|{\bf Y}|\geq 1\}\right)}\right).
\end{align*}
Since $E_2\cap E_3$ implies that $|{\bf Y}|\geq 1,$ we have
\begin{align*}
&\Pr(E_1||{\bf Y}|\geq 1)+\Pr\left(E_2\cap E_3||{\bf Y}|\geq 1\right)\\
=&\frac{1}{\Pr(|{\bf Y}|\geq 1)}\left({\Pr(E_1\cap \{|{\bf Y}|\geq 1\})}+{\Pr\left(E_2\cap E_3\right)}\right)\\
=&\frac{1}{\Pr(|{\bf Y}|\geq 1)}\left(\Pr(E_1)-{\Pr(E_1\cap \{|{\bf Y}|=0\})}\right.\\
&\left.+{\Pr\left(E_2\cap E_3\right)}\right)\\
\geq &\frac{1}{\Pr(|{\bf Y}|\geq 1)}\left(\Pr(E_1)-{\Pr(\{|{\bf Y}|=0\})}+{\Pr\left(E_2\cap E_3\right)}\right)\\
\geq &\frac{1}{\Pr(|{\bf Y}|\geq 1)}\left({\Pr(\{|{\bf Y}|\geq 1\})}-\epsilon_5\right)\\
= & 1-\frac{\epsilon_5}{\Pr(|{\bf Y}|\geq 1)}.
\end{align*}
Note that $\Pr(|{\bf Y}|\geq 1)$ is a positive constant since \textcolor{black}{the one time slot infection process} starting from the information source survives with non-zero probability. The theorem holds by choosing $\epsilon_5=\epsilon \Pr(|{\bf Y}|\geq 1).$

\begin{lemma}\label{lem:2survive}
The extinction probability of an one time slot infection process is smaller than the extinction probability of a binomial branching process $B(g_{\min},q_{\min}),$ i.e., $\forall v\in{\cal V},$
$$\rho_v<\rho.$$
\end{lemma}
\begin{proof}
As shown in Figure \ref{figure:CouplingSystems}, we construct a \emph{virtual source process} $Z^{(vs)}_l(T_{v}^{-\phi(v)})$ and a \emph{min-infection process} $Z^{(mi)}_l(T_{v}^{-\phi(v)})$ as auxiliary processes over the same tree topology where $Y_v^{(vs)}$ and $Y_v^{(mi)}$ are the binary numbers indicating whether node $v$ has been infected. Denote by $\rho^{(vs)}_v$ and $\rho^{(mi)}_v$ the extinction probabilities, respectively.

In the min-infection process,  infection spreads over edges with probability $q_{\min}.$ In the virtual source process, the probability that a node gets infected is
\begin{align*}
\Pr(Y_v^{(vs)}\!=\!1)&=\Pr(Y_v^{(mi)}\!=\!1)\!+\!\Pr(Y_v^{(mi)}=0)\!\cdot\!\frac{q_{uv}-q_{\min}}{1-q_{\min}}\\
&=q_{uv},
\end{align*}
i.e., for each node $u\in{\cal C}(v),$ $v$ tries to infects $u$ with probability $q_{\min}.$ If $v$ fails to infect $u,$ a \emph{virtual source} $v'$ tries to infect $u$ with probability $\frac{q_{vu}-q_{\min}}{1-q_{\min}}.$ Therefore, the virtual source process has the same distribution with the one time slot infection process.

We now couple the min-infection process and the virtual source infection process as follows:
\begin{itemize}
\item If $Y_v^{(mi)}=1,$ then $Y_v^{(vs)}=1.$
\item If $Y_v^{(mi)}=0,$ then $Y_v^{(vs)}=1$ with probability $\frac{q_{uv}-q_{\min}}{1-q_{\min}}.$
\end{itemize}
Since a node is more likely to get infected in the virtual source infection process, we obtain
\[
\rho^{(vs)}_v\leq \rho^{(mi)}_v.
\]
Recalling the one time slot infection process has the same distribution with the virtual source branching process, we obtain $\rho_v\leq \rho^{(mi)}_v,\forall v.$

In addition, the min-infection process has more children than the binomial branching process with the same infection probability for each children. It is obvious that the binomial branching process is more likely to die out, i.e., $\rho^{(mi)}_v<\rho.$

As a summary, we prove
\[
\rho_v<\rho.
\]
\end{proof}

\begin{lemma}\label{lem:morethann}
Assume $\exists \xi>0$ such that $\sigma^\tau_v<1-\xi,\forall v\in{\cal V}.$ Given any $\epsilon>0,$ there exists a constant $L^{\prime}$ such
that for any $L\geq L^{\prime},$
\[
\Pr\left(\bigcap_{i=1}^{L}0<Z^{\tau}_i(T_{v^*})\leq
n_0\right)\leq \epsilon
\]
\end{lemma}

\begin{proof}
 Follows the same argument of Lemma 7 in \cite{ZhuYin_12}, by choosing
$$L^{\prime}=\left\lceil\frac{\log \epsilon}{\log\left(1-\xi^{n_0}\right)}\right\rceil,$$ we obtain for any $L\geq L^{\prime}, \epsilon>0$
$$\Pr\left(\bigcap_{i=1}^{L}0<Z^{\tau}_i(T_{v^*})\leq n_0\right)\leq
\epsilon.$$
\end{proof}

\begin{lemma}\label{lem:NodeNumberIncrease}
For any $\epsilon>0,$ there exists a sufficiently large $t$ such that
\[
\Pr(E)\geq 1-\epsilon
\]
\end{lemma}

\begin{proof}
Note the binomial branching process $B(g_{\min},q_{\min})$ is a Galton-Watson (GW) process \cite{Har_63} which requires each node has an i.i.d offspring distribution. The previous result about the instability of the Galton-Watson process in Theorem 6.2 in \cite{Har_63} proves that the GW process either goes to infinity or goes to $0.$ If the GW process survives, the number of offsprings goes to infinity as the level increases. Therefore, for sufficiently long time, the survived binomial branching process will have a  sufficiently large number of offsprings at the lowest level. Since the one time slot infection process always has at least the same number of children as the binomial branching process, the survived one time slot infection process will have enough number of infected nodes at the lowest level as time increases. According to the unbiased property of the partial observation, after sufficiently long time, the probability that at least one infected node in the lowest level is observed goes to 1 asymptotically, i.e.,
\[
\Pr(E)\geq 1-\epsilon.
\]
\end{proof}
\begin{figure}[t!]
\begin{centering}
  \includegraphics[width=0.5\columnwidth]{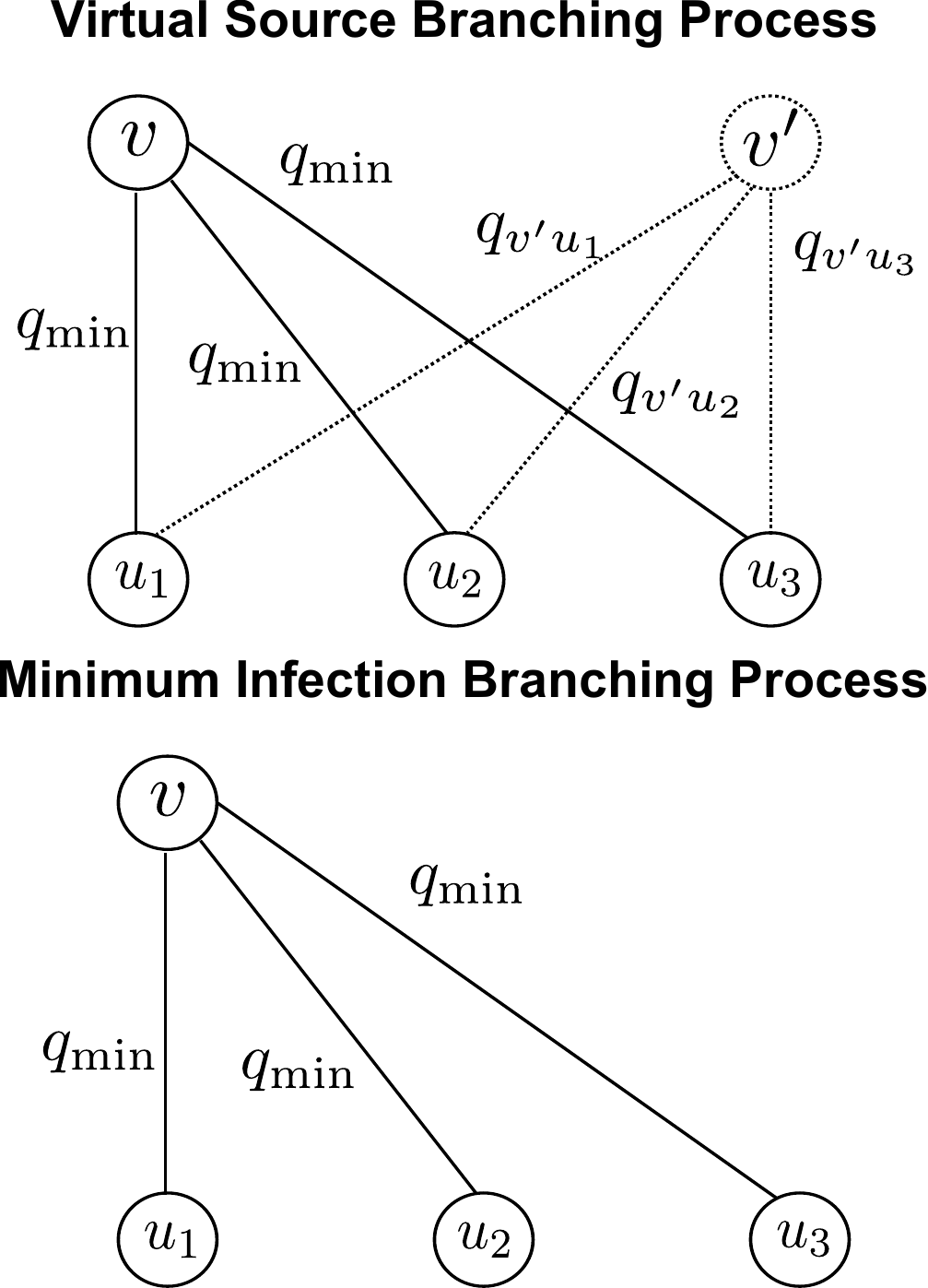}\\
  \caption{A Pictorial Description of the Two Auxiliary Processes in Lemma \ref{lem:2survive}}\label{figure:CouplingSystems}
  \end{centering}
\end{figure}

\section{Conclusion}
In this paper, we studied the problem of detecting the information source in a heterogeneous SIR model with sparse observations. We proved that the optimal sample path estimator on an infinite tree is a node with the minimum infection eccentricity with partial observations. With a fairly general condition, we proved that the estimator is within constant distance from the actual information source with a high probability with a sparse observation. Extensive simulation results showed our estimator outperforms other algorithms significantly. 
\balance
\bibliographystyle{IEEEtran}
\bibliography{information_source_main}

\end{document}